\documentclass[11pt,letterpaper]{article}
\usepackage[margin=1.0in]{geometry}
\usepackage{amsmath,amssymb,amsthm,array,tabularx,booktabs}
\usepackage[T1]{fontenc}
\usepackage{lmodern}
\usepackage{microtype}
\usepackage[table]{xcolor}
\usepackage[hidelinks]{hyperref}
\newtheorem{proposition}{Proposition}
\newtheorem{lemma}{Lemma}
\newtheorem*{theorem}{Theorem}

\newtheorem*{statement}{Statement}

\newcommand{\defeq}{\mathrel{\overset{\mathrm{def}}{=}}}
\newcommand{\eps}{\epsilon}
\newcommand{\ph}{\phi}

\usepackage{lmodern,amsmath,amssymb,tikz}

\usepackage[font=small,labelfont=bf]{caption}

\usetikzlibrary{arrows.meta}
\definecolor{searched}{HTML}{BCC3CB}
\definecolor{sweep}{HTML}{24749C}
\definecolor{axisgray}{HTML}{6C7580}
\definecolor{guidegray}{HTML}{D4D8DE}
\definecolor{frontier}{HTML}{243B4C}

\tikzset{
 pathaxis/.style={draw=axisgray,line width=.65pt,-{Stealth[length=2mm]}},
 oldground/.style={draw=searched,line width=5pt,line cap=butt},
 newground/.style={draw=sweep,line width=5pt,line cap=butt},
 guide/.style={draw=guidegray,line width=.4pt,dash pattern=on 2pt off 3pt},
 paneltitle/.style={font=\large},
 pathlabel/.style={font=\large},
 ticklabel/.style={font=\small},
 state/.style={font=\fontsize{8.5}{10}\selectfont},
 steplabel/.style={font=\small,text=sweep,fill=white,inner sep=1pt}
}

\title{On the Cyclic Assumption of the Cow-Path Search Algorithm}
\author{Yuan Ma \and Yiqun Lisa Yin\thanks{Corresponding author: \text{yiqun@alum.mit.edu}}}

\date{}
\begin{document}
\maketitle

\begin{abstract}
In the cow-path problem, a cow must find a goal lying at an unknown distance on one of
$w$ paths connected only at the origin, and performance is measured by competitive ratio. An
efficient randomized algorithm is designed in~\cite{KRT}, in which the cow visits the paths in a fixed cyclic order.
The algorithm is proved to be optimal for $w=2$ in~\cite{KRT}, and for all $w$ in~\cite{KMSY} with a claim that no algorithm does better than the best cyclic one. This note provides a detailed proof of that claim.
\end{abstract}

\section{Introduction} \label{section-intro}
In the $w$-lane cow-path problem, there are $w$ infinite paths from a common
origin and a goal lying at an unknown distance on one of these paths.
A cow (or robot) starts at the origin but does not know which path contains the goal or how far
the goal is, and it discovers the goal only upon reaching it. Its objective is to find the goal while
minimising the total distance travelled, and its performance is measured by competitive ratio. 

Baeza-Yates, Culberson and Rawlins~\cite{BCR} determined the optimal deterministic
competitive ratio, $1+2w^w/(w-1)^{w-1}$, which is $9$ for $w=2$. Kao, Reif and Tate~\cite{KRT} introduced a randomized algorithm with competitive ratio
\[
R_w=\min_{r>1}\Bigl\{1+\frac{2}{w}\cdot\frac{1+r+\cdots+r^{w-1}}{\ln r}\Bigr\},
\]
which is approximately $4.5911$ for $w=2$, proved its optimality
for $w=2$, and conjectured its optimality for $w>2$. Kao, Ma, Sipser and Yin~\cite{KMSY} subsequently proved this conjecture.

A common feature of these analyses is that the algorithm visits the paths in a fixed \emph{cyclic} order, which means that the algorithm
visits all the $w$ paths in its first $w$ visits and then repeatedly searches these paths in that order. Fact~4 in~\cite{KMSY} stated explicitly that it suffices to analyze these cyclic algorithms but left a detailed proof  
to a manuscript by us~\cite{MY},
which was never formally published. This note constructs a complete proof from some surviving hand-written notes.

\section{Preliminaries}

In this section, we formally introduce terminology and notation. Appendix~\ref{section-notation} will list all our terms and notation for ease of reference.
   
\subsection{Algorithms}
A search \emph{algorithm} is represented by a sequence
\[
X=\bigl((P_t,\ell_t)\bigr)_{t\ge1},
\]
fixed in advance and independent of the goal's location, where $P_t$ and  $\ell_t$ are  the path and length of step $t$, respectively. 
To avoid any ambiguity, a \emph{step} in an algorithm sweeps a path from the origin to
a finite distance and then returns to the origin.  

We write $f_t(P)$ for path
$P$'s {\it frontier}, defined to be the distance to which it has been swept when step $t$ is about to
begin. Note that $f_{t+1}(P_t)=\ell_t$.

Step $t$ costs $2\ell_t$ in travelled distance, out and back---with one exception, where the goal 
located at some distance $h$ from the origin is found during this step. 
On that step the cow stops at the goal on the outward leg and does not return, so step $t=k$ costs $h$ rather
than $2\ell_k$, and the algorithm ends. 
The total distance travelled is $2\sum_{i<k}\ell_i+h$, the sum
running over steps $1,\ldots,k-1$ and excluding step $k$ itself.

Throughout, $w\ge2$; for $w=1$ the cow simply walks down the single path and the ratio is $1$. Without loss of generality, we only need to work on algorithms with a finite expected competitive ratio, because such an algorithm was constructed in~\cite{KRT} already. 
This requires that the cow will not continue on a path indefinitely (i.e., $\ell_t$ is finite for all $t$) and it will search every
path without a constant upper bound (i.e., each path's sweep lengths are unbounded): otherwise the goal may never be found, leading to an infinite expected competitive ratio.

\subsection{The distribution and competitive ratio}

At the highest level, Yao's principle~\cite{Yao} allows us to lower-bound the worst-case expected ratio
of every randomized algorithm by proving a lower bound on the expected ratio of every deterministic algorithm
under a fixed input distribution. The goal's distribution is that specified in~\cite[p.~70]{KRT}, where
the path containing the goal is chosen uniformly from the $w$ paths, and the distance $x$ along
that path has density
\[
f_{\eps,w}(x)=
\begin{cases}
\eps x^{-(1+\eps)},&x\ge1,\\
0,&\text{otherwise}.
\end{cases}
\]
Throughout, $\eps>0$ is fixed.

Let $H$ denote the per-path measure, with density
\[
dH(h)=\frac{\eps}{w}h^{-(1+\eps)}\,dh,\qquad h\ge1.
\]
Thus each path carries mass $1/w$. If the goal is found at step $k$, the competitive ratio is
\[
1+\frac{2\sum_{i<k}\ell_i}{h},
\]
where the leading $1$ accounts for the final trip to the goal and the remaining term accounts for
unsuccessful round trips. 

Deleting search steps that do not cover any new ground, we can also assume
\begin{equation}\label{feasible}
\ell_t>f_t(P_t)\qquad\text{for every }t.
\end{equation}
This way, the outward leg retraces the already swept portion
$[0,f_t(P_t)]$, and only $(f_t(P_t),\ell_t] = (f_t(P_t),f_{t+1}(P_t)]$ is searched for the first time. 
Under this assumption (otherwise the integrals below would be ill-defined), summing over the goal locations gives
\begin{equation}\label{ratio}
\text{expected competitive ratio of }X
=\sum_k\int_{f_k(P_k)}^{\ell_k}\frac{2\sum_{i<k}\ell_i+h}{h}\,dH(h).
\end{equation}

Each unit of unsuccessful travel adds $1/h$ to the ratio. For a path searched to frontier $x$,
define
\begin{equation}\label{hazard}
\ph(x)\defeq\int_x^\infty\frac{dH(h)}{h}
=\frac{\eps}{w(1+\eps)}x^{-(1+\eps)}.
\end{equation}

For ease of notation, define
\begin{equation}\label{constants}
\alpha\defeq1+\eps,\qquad \beta\defeq(1+\alpha)^{1/\alpha},\qquad
\kappa\defeq\frac{\eps}{w(1+\eps)}.
\end{equation}
Then
\begin{equation}\label{small-phi}
\ph(x)=\kappa x^{-\alpha}.
\end{equation}
Also,
\begin{equation}\label{beta}
\beta^\alpha=1+\alpha=2+\eps,
\end{equation}
so $\beta>1$, and $\beta\to2$ as $\eps\to0$.

\subsection{States, remaining weights, and contributions}
The \emph{state} of the search is defined by the collection of all paths' frontiers plus the location of the cow. 
By default, the cow is understood to be at the origin, when not specified explicitly, and so a state can be understood to be the collection of all frontiers of the paths. With such an understanding, we use capital letter $F$ to denote a state. 

In particular, $F$ assigns to each path $P$ a frontier $F(P)$. We write $F_t$ for the state when step $t$
is about to be taken, so that $F_t(P)=f_t(P)$. 

We call $\ph(x)$ the \emph{remaining weight} on the path: the probability mass beyond $x$ normalized by $h$,
reflecting that the same extra travel is more costly relative to a nearby goal.
The function $\ph$ is positive, strictly decreasing, and smooth; searching farther reduces this
remaining weight. The constant $\kappa$ is an overall positive factor in every comparison below
and cancels; only the exponent $\alpha$ matters.

For a given state $F$, summing the remaining weights over all paths gives us the \emph{total remaining weight}
\begin{equation}\label{Phi}
\Phi(F)\defeq\sum_P\ph\bigl(F(P)\bigr), \mbox{ and}
\end{equation}

\begin{equation}\label{Phi_t}
\Phi_t\defeq\Phi(F_t).
\end{equation}
\noindent
The total remaining weight depends only on the current frontiers and is
unchanged by relabelling paths. Note that under such definitions, $\Phi_t$ and $\Phi_{t+1}$ are total remaining weights before
and after step $t$, respectively.

We next define the \emph{contribution} for a full round-trip step $t$ that travels the distance $2\ell_t$ for goal locations that remain
un-searched after step $t$. The total remaining weight after step $t$ is $\Phi_{t+1}$, and step $t$ \emph{contributes}
$2\ell_t\Phi_{t+1}$ to the expected future cost. Summing over all steps of algorithm $X$ starting from $F$ gives us the following:

\begin{equation}\label{objective}
E_F(X)\defeq2\sum_t\ell_t\Phi_{t+1}.
\end{equation}
Notice that the above quantity is unchanged by
permuting the labels of the paths, since $\Phi$ is a symmetric sum and all paths carry the same $\ph$ and
the same dynamics. In addition, $\ell_t$ affects both the travel distance now and the post-step
remaining weight through $\ph(\ell_t)$: a longer sweep costs more travel now but leaves less remaining weight.

For the initial state $F_0$, in which $F_0(P)=1$ for every path $P$, interchanging the nonnegative
terms in~\eqref{ratio} gives the expected competitive ratio
\begin{equation}\label{oneplus}
1+E_{F_0}(X).
\end{equation}
Minimising the expected competitive ratio is therefore equivalent to minimising $E_{F_0}$,
which we will work on in the rest of the paper.

Equation~\eqref{objective} is a sum over all steps, so comparisons of algorithms modified from a certain step onwards must account for their
\emph{continuations}, defined to be the remaining steps of an algorithm starting from a certain state. By~\eqref{objective} the contribution of a step depends on the past only through the current
state. Identical continuations started at the same state have equal contributions, and the
infimum of the contribution of a continuation depends only on that state.

Finally, since we only need to deal with algorithms with finite competitive ratios, we assume in all our arguments that the algorithms we will modify all have finite contributions.


\section{Proposition 1: prioritize a least-frontier path}
\begin{proposition}[least-frontier first]\label{least}
If  algorithm $X$ sweeps a path whose frontier
is strictly larger than the least frontier at some step, then $X$ can be modified into another algorithm $X^*$ such that starting from that state
and that step the contribution of $X^*$ is strictly less than that of $X$.
\end{proposition}
\begin{proof}
Suppose that in $X$, starting from state $F$, step $t$ first sweeps path $B$, whose frontier is $b$, to a distance
$y>b$, while another path $A$ has frontier $a<b$.

Construct $X^*$ from $X$ as follows. Make no change before step $t$, but let step $t$ sweep $A$ to $y$, which is larger than $b>a$, and then follow
the remaining steps of $X$ with the labels $A$ and $B$ interchanged, keeping all search distances unchanged.

To make $X^*$ satisfy condition (\ref{feasible}), we skip any prescribed sweep of $B$ whose searched distance is at most $b$ (note the relabelled continuation of $X$ starts with $A$ at $y$ and $B$
at $a$, whereas $X^*$ starts its continuation with $A$ at $y$ and $B$ already at $b>a$). The prescribed search distances on $B$ are strictly increasing
and unbounded, and so only finitely many sweeps are skipped.

\bigskip

\begin{figure}[!ht]
\centering
\begin{minipage}{0.9\textwidth}
\resizebox{0.9\linewidth}{!}{%
\begin{tikzpicture}[x=.60cm,y=1.15cm]
\foreach \offset/\title/\aa/\ab/\ea/\eb/\statea/\stateb in {
 0/{$X$}/2/4/2/7/a/y,
 9.3/{$X^*$ before $A\leftrightarrow B$}/2/4/7/4/y/b,
 18.6/{$X^*$ after $A\leftrightarrow B$}/4/2/4/7/b/y}{
 \begin{scope}[xshift=\offset*.60cm]
  \node[paneltitle] at (3.6,2.05) {\title};
  \foreach \z/\lab in {2/a,4/b,7/y}{
   \draw[guide] (\z,-.15)--(\z,1.48);
   \node[ticklabel] at (\z,-.35) {$\lab$};
  }
  \node[ticklabel] at (0,-.35) {$0$};
  \foreach \h/\name/\old/\end in {1.15/A/\aa/\ea,0/B/\ab/\eb}{
   \draw[pathaxis] (0,\h)--(7.85,\h);
   \node[pathlabel] at (-.55,\h) {$\name$};
   \draw[oldground] (0,\h)--(\old,\h);
   \ifdim\end pt>\old pt
    \draw[newground] (\old,\h)--(\end,\h);
   \fi
   \fill[frontier] (\end,\h) circle[radius=1.6pt];
  }
  \node[state] at (3.6,-.83) {Post-step-$t$ frontiers: $(\statea,\stateb)$};
 \end{scope}
}
\draw[oldground] (6.0,-1.43)--(6.8,-1.43);
\node[anchor=west,font=\small] at (7.0,-1.43) {Previously searched};
\draw[newground] (15.1,-1.43)--(15.9,-1.43);
\node[anchor=west,font=\small] at (16.1,-1.43) {sweep in step $t$};
\end{tikzpicture}%
}
\caption{Comparison of $X$ and $X^*$ in the proof of Proposition 1.}
\label{fig:prop1}
\end{minipage}
\end{figure}
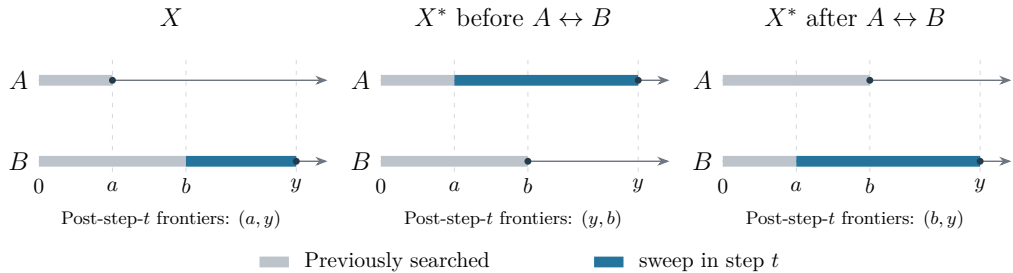

For $X$, let $C_1,C_2$ be the contributions of step $t$
and continuation, respectively.  Let $C^*_1,C^*_2$ be the corresponding contributions of $X^*$. It suffices
to show
\[
C^*_1<C_1,\qquad C^*_2\le C_2.
\]

For step $t$, $X$ leaves $A$ at $a$ and advances $B$ to $y$, while $X^*$ advances $A$ to $y$
and leaves $B$ at $b$. All other frontiers are unchanged. The post-step remaining weights are therefore
\[
\Phi(F)-\ph(b)+\ph(y)
\quad\text{and}\quad
\Phi(F)-\ph(a)+\ph(y),
\]
for $X$ and $X^*$ respectively. Both steps have length $y$, so
\begin{equation}\label{least-gap}
\begin{aligned}
C_1-C^*_1
&=2y\bigl[\Phi(F)-\ph(b)+\ph(y)\bigr]
  -2y\bigl[\Phi(F)-\ph(a)+\ph(y)\bigr]\\
&=2y\bigl[\ph(a)-\ph(b)\bigr]>0,
\end{aligned}
\end{equation}
because $a<b$ and $\ph$ is strictly decreasing.

For the continuation, interchanging the labels $A$ and $B$ preserves the contribution of $X$,
since $\Phi$ is symmetric in the path frontiers. Compare this relabelled continuation with that of
$X^*$. Until the first prescribed sweep of $B$ beyond $b$, the frontier of $B$ in $X^*$ is at least as large as that of $X$,
while all other frontiers agree at corresponding steps. From that sweep onward, the states coincide.
Consequently, every retained step has the same length and no larger post-step remaining weight, and so it contributes
no more. The skipped steps contribute nothing. Thus
\[
C^*_2\le C_2.
\]
Combining the comparisons gives
\[
E_F(X^*)=C^*_1+C^*_2<C_1+C_2=E_F(X). \qedhere
\]
\end{proof}

\section{Proposition 2: each sweep advances its path by strictly more than a factor \texorpdfstring{$\beta$}{beta}}

In what follows, let $t$ be a step in algorithm $X$, let $t^+$ be the next step
sweeping $P_t$, and define
\begin{equation}\label{window}
W_t\defeq\sum_{t\le j<t^+}\ell_j, \qquad \widehat W_t\defeq \sum_{t<j<t^+}\ell_j.
\end{equation}

\begin{lemma}[stationarity]\label{stationarity}
Unless
\begin{equation}\label{stationarity-eq}
-\phi'(\ell_t) W_t = \Phi_{t+1},
\end{equation}
replacing $\ell_t$ by a suitable $y$ (and leaving all other steps unchanged) yields an algorithm
contributing strictly less than $X$. 
\end{lemma}
\begin{proof}

Define
\[
\widehat{\Phi}_{t+1}\defeq \Phi_{t+1}-\ph(\ell_t),
\]
the total remaining weight (post-step-$t$) associated with all paths other than $P_t$.

We vary only the length of step $t$, replacing $\ell_t$ by $y$, while leaving all other steps unchanged.
The feasible range is 
$f_t(P_t) <y<\ell_{t^+}$, so $y=\ell_t$ is an interior point and can be varied in either direction.

Step $t$ now contributes $2y[\widehat{\Phi}_{t+1}+\ph(y)]$. During the steps strictly between $t$ and $t^+$, path $P_t$
remains at frontier $y$, so its associated contribution is $2\widehat W_t\ph(y)$. At step $t^+$, its frontier becomes
$\ell_{t^+}$. Thus the total contribution is
\begin{equation}\label{vary}
C(y)=K+2y\bigl[\widehat{\Phi}_{t+1}+\ph(y)\bigr]+2\widehat W_t\ph(y),
\end{equation}
where $K$ collects all contributions unaffected by $y$.

Differentiating~\eqref{vary},
\[
C'(y)=2\widehat{\Phi}_{t+1}+2\ph(y)+2(y+\widehat W_t)\ph'(y).
\]
At $y=\ell_t$, we have $\widehat{\Phi}_{t+1}+\ph(\ell_t)=\Phi_{t+1}$ and $\ell_t+\widehat W_t=W_t$, and hence
\[
C'(\ell_t)=2\Phi_{t+1}+2\ph'(\ell_t)W_t.
\]
If this derivative is positive, a small decrease of $\ell_t$ improves the algorithm; if it is negative, a small increase
does. Thus an algorithm admitting no one-step variation improvement must satisfy
\[
\Phi_{t+1}=-\ph'(\ell_t)W_t. \qedhere
\]
\end{proof}

Lemma~\ref{stationarity} says that each step must balance its current sweep length against the accumulated
sweep lengths until that path is swept again. 
We call $X$ \emph{stationary} if equality~\eqref{stationarity-eq} holds at every step. 

\begin{proposition}[advance factor]\label{advance}
For a stationary algorithm, if
\[
\ell_t\le\beta f_t(P_t)
\]
at some step $t$, then deleting step $t$ yields an algorithm with strictly less contribution.
\end{proposition}
\begin{proof}
Define
\[
p\defeq f_t(P_t). 
\]
Deleting step $t$ saves its contribution $2\ell_t\Phi_{t+1}$. However, until step $t^+$, path $P_t$
remains at frontier $p$ rather than $\ell_t$, adding contribution
\[
2[\ph(p)-\ph(\ell_t)]\widehat W_t.
\]
Note that condition (\ref{feasible}) holds since the next sweep of this path reaches $\ell_{t^+}>\ell_t>f_t(P_t)$
and that every path remains swept without bound. Thus it suffices to show that the additional
contribution is smaller than the saving.

Since $X$ is stationary, by Lemma~\ref{stationarity}, the following equality holds:
\[
\Phi_{t+1}=-\ph'(\ell_t)W_t.
\]
Since $-\ph'(\ell_t)=\alpha\ph(\ell_t)/\ell_t$, we obtain
\begin{equation}\label{advance-balance}
\alpha\ph(\ell_t)\widehat W_t
=\ell_t\Phi_{t+1}-\alpha\ell_t\ph(\ell_t)
<\ell_t\Phi_{t+1}.
\end{equation}

Now let $\rho=\ell_t/p$. Since $1<\rho\le\beta$ and $\beta^\alpha=1+\alpha$, we have
\begin{equation}\label{advance-hazard}
\ph(p)-\ph(\ell_t)
=\ph(\ell_t)(\rho^\alpha-1)
\le\alpha\ph(\ell_t).
\end{equation}
Combining inequalities (\ref{advance-balance}) and (\ref{advance-hazard})  gives
\begin{equation}\label{advance-deletion}
[\ph(p)-\ph(\ell_t)]\widehat W_t
\le\alpha\ph(\ell_t)\widehat W_t
<\ell_t\Phi_{t+1}.
\end{equation}
The deletion penalty is therefore strictly smaller than the saving, so deleting step $t$ yields
an algorithm contributing strictly less.
\end{proof}

\section{Proposition 3: step lengths are nondecreasing}
In this section, we analyze the effect of exchanging the search lengths of two consecutive sweeps, especially when the lengths decrease.
By relabelling the paths, we can reduce the comparison to
the two exchanged steps. 

Let $X$ be an algorithm with a finite contribution. Suppose step $k$ sweeps path $A$ from $a$ to $u$,
and step $k+1$ sweeps a distinct path $B$ from $b$ to $v$. 
Define $X^*$ as follows:
\begin{itemize}
\item Steps $1,\ldots,k-1$ are the same as in $X$.
\item Step $k$ sweeps $A$ from $a$ to $v$, and step $k+1$ sweeps $B$ from $b$ to $u$.
\item For every step $j\ge k+2$, use the same sweep length as $X$, with the path labels $A$ and $B$ interchanged.
\end{itemize}

\begin{lemma}[exchange of adjacent sweep lengths]\label{exchange}
For any algorithm X and its corresponding $X^*$, if $\min(u,v) > \max(a,b)$, then
\begin{equation}\label{exchange-difference} 
E_{F_0}(X)-E_{F_0}(X^*)=2\bigl[(u-v)\ph(b)+v\ph(u)-u\ph(v)\bigr].
\end{equation}
\end{lemma}
\begin{proof}
First, the two algorithms agree before step $k$.  Secondly, after step $k+1$, paths $A,B$ have
frontiers $u,v$ in $X$ and $v,u$ in $X^*$ and all other frontiers agree. Interchanging the labels $A,B$ after step $k+1$
therefore makes the two states identical. So the continuations of $X$ and $X^*$ are the same.

\bigskip

\begin{figure}[!ht]
\centering
\begin{minipage}{0.9\textwidth}
\resizebox{0.9\linewidth}{!}{%
\begin{tikzpicture}[x=.60cm,y=1.15cm]
\foreach \offset/\title/\aa/\ab/\ea/\eb/\sa/\sb/\statea/\stateb in {
 0/{$X$}/1.5/3/7/5/k/{k+1}/u/v,
 9.3/{$X^*$ before $A\leftrightarrow B$}/1.5/3/5/7/k/{k+1}/v/u,
 18.6/{$X^*$ after $A\leftrightarrow B$}/3/1.5/7/5/{k+1}/k/u/v}{
 \begin{scope}[xshift=\offset*.60cm]
  \node[paneltitle] at (3.6,2.05) {\title};
  \foreach \z/\lab in {1.5/a,3/b,5/v,7/u}{
   \draw[guide] (\z,-.15)--(\z,1.53);
   \node[ticklabel] at (\z,-.35) {$\lab$};
  }
  \node[ticklabel] at (0,-.35) {$0$};
  \foreach \h/\name/\old/\end/\step in {1.15/A/\aa/\ea/\sa,0/B/\ab/\eb/\sb}{
   \draw[pathaxis] (0,\h)--(7.85,\h);
   \node[pathlabel] at (-.55,\h) {$\name$};
   \draw[oldground] (0,\h)--(\old,\h);
   \draw[newground] (\old,\h)--(\end,\h);
   \fill[frontier] (\end,\h) circle[radius=1.6pt];
   \node[steplabel] at ({(\old+\end)/2},{\h+.26}) {step $\step$};
  }
  \node[state] at (3.6,-.83) {Post-step-$(k+1)$ frontiers: $(\statea,\stateb)$};
 \end{scope}
}
\draw[oldground] (4.5,-1.43)--(5.3,-1.43);
\node[anchor=west,font=\small] at (5.5,-1.43) {Previously searched};
\draw[newground] (13.6,-1.43)--(14.4,-1.43);
\node[anchor=west,font=\small] at (14.6,-1.43) {sweeps in steps $k$ and $k+1$};
\end{tikzpicture}%
}
\caption{Comparison of $X$ and $X^*$ in Lemma 2. After relabelling, the frontiers of $X^*$ match those of $X$, allowing the same continuation of the algorithm.}
\label{fig:lemma2}
\end{minipage}
\end{figure}
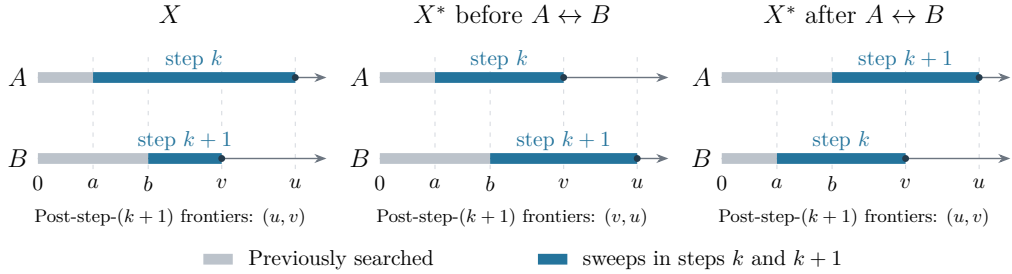

In sum, only steps $k$ and $k+1$ contribute to any difference. Moreover, since $\min(u,v) > \max(a,b)$, both algorithms satisfy condition (\ref{feasible}) and so we can use our established formulae to compute the difference of $E_{F_0}$ for the two algorithms. 
Write $\Phi_{k+1}^*$ and $\Phi_{k+2}^*$ for the corresponding remaining weights in $X^*$. Since the states after
the two sweeps differ only by relabelling,
\[
\Phi_{k+2}^*=\Phi_{k+2}.
\]
Also,
\[
\Phi_{k+1}=\Phi_{k+2}+\ph(b)-\ph(v),\qquad
\Phi_{k+1}^*=\Phi_{k+2}+\ph(b)-\ph(u).
\]
Thus
\[
\begin{aligned}
\tfrac12\bigl[E_{F_0}(X)-E_{F_0}(X^*)\bigr]
&=u\Phi_{k+1}+v\Phi_{k+2}-v\Phi_{k+1}^*-u\Phi_{k+2}\\
&=u\bigl[\ph(b)-\ph(v)\bigr]-v\bigl[\ph(b)-\ph(u)\bigr]\\
&=(u-v)\ph(b)+v\ph(u)-u\ph(v).
\end{aligned} \qedhere
\]
\end{proof}
Note that~\eqref{exchange-difference} involves only $\ph(b)$, the frontier of the path swept second, and not $\ph(a)$.

\begin{proposition}[monotonicity]\label{monotonicity}
Let $X$ be an algorithm started from $F_0$ such that no algorithm
from $F_0$ contributes strictly less than $X$. Then the step lengths of $X$ are nondecreasing.
\end{proposition}
\begin{proof}
Because no algorithm contributes strictly less than $X$, no local modification improves it either.
Also, $X$ has finite contribution, as otherwise an algorithm sweeping the paths cyclically with a fixed geometric ratio larger than 1 will be an improvement. Hence, by Proposition~\ref{least}, Lemma~\ref{stationarity} and Proposition~\ref{advance} in contrapositive form,
\begin{itemize}
\item every step sweeps a path of least current frontier;
\item $X$ is stationary;
\item every step satisfies $\ell_t>\beta f_t(P_t)$.
\end{itemize}
Suppose the step lengths are not nondecreasing in $X$, and choose an adjacent decrease (i.e., $\ell_{k+1} < \ell_{k}$). Write
\[
u\defeq \ell_k,\qquad v\defeq \ell_{k+1},\qquad \mbox{ so } v<u.
\]
Steps $k$ and $k+1$ sweep different paths, because two consecutive sweeps of the same path must strictly
increase its frontier. Let
\[
A\defeq P_k,\qquad B\defeq P_{k+1},\qquad a\defeq f_k(A),\qquad b\defeq f_k(B)=f_{k+1}(B).
\]
Since step $k$ sweeps a path of least frontier,
\begin{equation}\label{frontier-order}
a\le b.
\end{equation}
Also, Proposition~\ref{advance} applied to step $k+1$ gives
\begin{equation}\label{advance-pair}
v>\beta b.
\end{equation}
The above two inequalities imply $\min(u,v)>\max(a,b)$, and so we can use Lemma~\ref{exchange} now.
Substituting $\ph(x)=\kappa x^{-\alpha}$ into~\eqref{exchange-difference} and cancelling $2\kappa>0$,
the exchange strictly lowers the contribution if and only if
\begin{equation}\label{exchange-positive}
(u-v)b^{-\alpha}+vu^{-\alpha}-uv^{-\alpha}>0.
\end{equation}

To prove~\eqref{exchange-positive}, we may assume $b=1$: dividing $u,v,b$ by $b>0$ multiplies
its left-hand side by the positive factor $b^{\alpha-1}$, preserving its sign. We then have $u>v>\beta$. Define
\[
g(u,v)\defeq (u-v)+vu^{-\alpha}-uv^{-\alpha}.
\]
We must show $g(u,v)>0$.

For fixed $v$, we have $g(v,v)=0$. Thus it is enough to show that $g$ increases as $u$ moves above $v$.
Differentiating with respect to $u$,
\[
\frac{\partial g}{\partial u}=1-v^{-\alpha}-\alpha vu^{-\alpha-1}.
\]
The role of Proposition~\ref{advance} is now clear: $u>v$ alone is not enough. At the diagonal,
\[
\frac{\partial g}{\partial u}(v,v)=1-(1+\alpha)v^{-\alpha},
\]
which is positive exactly when $v>\beta$, since $\beta^\alpha=1+\alpha$.

Moreover, for $u>v$, $u^{-\alpha-1}<v^{-\alpha-1}$, so
$\alpha vu^{-\alpha-1}<\alpha v\,v^{-\alpha-1}=\alpha v^{-\alpha}$. Therefore,
\[
\frac{\partial g}{\partial u}
=1-v^{-\alpha}-\alpha vu^{-\alpha-1}
>1-(1+\alpha)v^{-\alpha},
\]
which is already proved to be positive in the diagonal case.
Thus~\eqref{exchange-positive} is proved. So Lemma~\ref{exchange} gives an algorithm $X^*$ that contributes strictly less than $X$, leading to
a contradiction.
\end{proof}

\section{Main result}

The theorem below shows that the lower bound established in~\cite{KMSY} for cyclic algorithms is a lower bound for all
algorithms. Therefore, the cyclic assumption may be dropped, as claimed in Fact~4 of~\cite{KMSY}.

\begin{theorem} \label{main}
The least expected competitive ratio over all algorithms started from the state $F_0$ can be attained by a cyclic algorithm. 
\end{theorem}

\begin{proof}

First, we should not take for granted that an optimal algorithm exists: there could be infinitely many algorithms that approach optimality without ever reaching it.
The proof that an optimal algorithm exists for our chosen distribution is rather technical and left to Appendix~\ref{section-existence}. 
Here, we simply start with an optimal algorithm $X$ and prove that $X$ can be converted into a cyclic algorithm.  

By Propositions~1 and~3, every step of $X$ sweeps a path of least frontier and the sweeping lengths satisfy
$\ell_1\le\ell_2\le\cdots$. Note that each path's frontier is the length of its last
sweep, or $1$ if unswept. We construct an algorithm $Y$ from $X$ by breaking every least-frontier tie in favor of the least recently swept path and
interchanging the tied labels from that step on; this preserves the multiset frontiers and the contribution. 

The resulting algorithm $Y$ is cyclic. Since every $\ell_t>1$, the first $w$ steps sweep distinct paths.
Inductively, suppose steps $t-w+1,\dots,t$ sweep distinct paths. The least frontier is then $\ell_{t-w+1}$, held by $P_{t-w+1}$,
which is also the least recently swept path. Step $t+1$ therefore sweeps $P_{t-w+1}$, continuing the cyclic order. The cyclic algorithm $Y$ attains the least
competitive ratio.
\end{proof}

\section{Discussion}

Professor Allan Borodin recently pointed out to us that Schuierer~\cite{Schuierer} introduced another probability distribution to give an alternative and simpler proof of the lower bound established in~\cite{KMSY}, but his proof relied on the cyclic assumption. 
It is natural to ask if one can prove an optimal algorithm must be cyclic under his
distribution (a truncated $H_D$ where density is 0 beyond a bound $D$). However, our arguments use the facts
that $\ph$ is positive, strictly decreasing, and smooth, whereas under $H_D$ the function $\ph$ vanishes at
and beyond $D$. Moreover, each path is swept to $D$ in finitely many steps in~\cite{Schuierer}, and
algorithm modifications like ours become difficult. 
Nevertheless, numerics
for $w=3$ suggest the optimal algorithm may still be cyclic.

\section*{Acknowledgements}
We thank Professor Allan Borodin for his recent inquiry about the cyclic assumption, which prompted this research note. His inquiry made us realize there is a continued interest in our proof, which was obtained in 1994 but probably not typed up. 

With efforts, we were able to find only
some hand-written claims and proof sketches that survived two cross-country moves. The detailed proof presented here is constructed from those notes, with the aid of AI, especially Claude Opus 5, in September 2026.

We thank Professor Dan Kleitman again for his invaluable discussions, which helped us come up with the original proof when we conducted this research at MIT in the 1990s.

\definecolor{ink}{HTML}{000000}
\definecolor{change}{HTML}{FFFFFF}
\definecolor{clarify}{HTML}{FFFFFF}
\definecolor{muted}{HTML}{000000}

\setlength{\emergencystretch}{2em}
\renewcommand{\arraystretch}{1.1}
\newcolumntype{L}[1]{>{\raggedright\arraybackslash}p{#1}}
\newcolumntype{Y}{>{\raggedright\arraybackslash}X}
\newcommand{\heading}[1]{\vspace{7pt}{\large\bfseries\color{ink}#1}\par}
\newcommand{\highlight}[3]{\noindent\colorbox{#1}{\parbox{\dimexpr\linewidth-2\fboxsep\relax}{\textbf{#2}\par #3}}\par}

\appendix

\section{Existence of an optimal algorithm} \label{section-existence}

We begin by defining the following infimum, as opposed to a minimum, because we do not want to assume the infimum is reachable by an algorithm: 
\[
V(F)\defeq\inf_{X\text{ started at }F}E_F(X).
\]
The proof in this section is rather technical but independent of any of the propositions proved in earlier sections. In particular, the following general statement has nothing to do with cyclicity.

\begin{statement}
\noindent For $w\ge2$, there exists an algorithm $X$ starting from $F_0$ such that $E_{F_0}(X)=V(F_0)$.
\end{statement}

\begin{proof}
Write $V_0=V(F_0)$, which is finite because sweeping the paths cyclically with geometrically increasing lengths gives finite contribution. Choose algorithms $X^{(1)},X^{(2)},\ldots$ from $F_0$ such that
\[
E_{F_0}(X^{(n)})\le V_0+1,
\qquad E_{F_0}(X^{(n)})\longrightarrow V_0.
\]
Denote step lengths and swept paths of $X^{(n)}$ by $\ell_t^{(n)}$ and $P_t^{(n)}$. We will extract a limiting sequence and show that it yields an algorithm attaining $V_0$.

\medskip
\noindent\textbf{Part 1: bound each step length so that convergent subsequences can be extracted.}

We show that, for every fixed $t$, there is a finite constant $B_t$, independent of $n$, such that $\ell_t^{(n)}\le B_t$. The proof proceeds by induction on $t$. Since every term of (\ref{objective}) is nonnegative,
\begin{equation}
2\ell_t^{(n)}\Phi_{t+1}^{(n)}\le V_0+1
\qquad\text{for all }n,t.\label{eq:term-bound}
\end{equation}
\textit{Base case: $t=1$.} The $w-1$ paths other than $P_1^{(n)}$ still have frontier $1$ after the first step. Thus $\Phi_2^{(n)}\ge(w-1)\kappa$. Define
\[
B_1\stackrel{\mathrm{def}}{=}\frac{V_0+1}{2(w-1)\kappa}, 
\]
where the denominator is positive as  $w  \ge 2$. Therefore, $\ell_1^{(n)}\le B_1$ by \eqref{eq:term-bound}. 

\noindent\textit{Induction step:} Now suppose $\ell_j^{(n)}\le B_j$ for all $j<t$ and all $n$, and define
\[
B_t\stackrel{\mathrm{def}}{=}\frac{(V_0+1)B^\alpha}{2(w-1)\kappa},
\]
where
\[
B\stackrel{\mathrm{def}}{=}\max\{1,B_1,\ldots,B_{t-1}\}.
\]
After step $t$, every path other than $P_t^{(n)}$ has frontier at most $B$: its frontier is either $1$ or a length from an earlier step. Since $\phi$ is decreasing, each of these $w-1$ paths contributes at least $\phi(B)=\kappa B^{-\alpha}$, and so $\Phi_{t+1}^{(n)}\ge(w-1)\kappa B^{-\alpha}>0$. 
Hence, $\ell_t^{(n)}\le B_t$ by \eqref{eq:term-bound}, thereby completing the induction. 

\medskip
\noindent\textbf{Part 2: extract a limiting sequence whose contribution is at most $V_0$.}

Among the algorithms $X^{(n)}$, some path label at step $1$ occurs infinitely often. Select a subsequence on which this label is constant, and call it $P_1$. Since the lengths $\ell_1^{(n)}$ lie in $[1,B_1]$, the Bolzano--Weierstrass theorem gives a further subsequence on which they converge to some $\ell_1$.

Repeat these two selections for step $2$, then step $3$, and so on, each time working within the preceding subsequence. For each step $t$, define $P_t$ to be the constant path label and $\ell_t$ to be the limit of the step lengths. Further selections preserve the labels and limits established earlier. This defines the following algorithm:
\[
X\stackrel{\mathrm{def}}{=}((P_t,\ell_t))_{t\ge1}.
\]
We will verify in Part 3 that every path is swept without bound and every step strictly advances its path's frontier.

Fix $T$ and consider the subsequence obtained after the selections for step $T$. Along this subsequence, the first $T$ path labels agree with those of $X$, and the corresponding lengths converge. The frontiers therefore also converge, giving $\ell_t\ge f_t(P_t)$ for $t\le T$. By continuity of $\phi$,
\begin{equation}
\ell_t^{(n)}\Phi_{t+1}^{(n)}
\xrightarrow{n\to\infty}\ell_t\Phi_{t+1}
\qquad(t\le T).\label{eq:term-convergence}
\end{equation}
Thus, taking limits along this subsequence,
\[
\begin{aligned}
2\sum_{t\le T}\ell_t\Phi_{t+1}
&=\lim_{n\to\infty}2\sum_{t\le T}\ell_t^{(n)}\Phi_{t+1}^{(n)}
&&\text{(since $T$ is finite)}  \\
&\le\lim_{n\to\infty}E_{F_0}(X^{(n)}) \\
& = V_0
\end{aligned}
\]
where the inequality holds because the omitted terms are nonnegative. Since this holds for every $T$, taking $T\to\infty$ gives
\[
2\sum_{t\ge1}\ell_t\Phi_{t+1}\le V_0.
\]

\medskip
\noindent\textbf{Part 3: ensure unbounded search and conclude optimality.}

First, every path is swept without bound in $X$ (so that the resulting algorithm always finds the goal). Otherwise, some path would have frontier at most $M$ at every step, giving $\Phi_{t+1}\ge\kappa M^{-\alpha}$. Since $\ell_t\ge1$,
\[
2\sum_{t\ge1}\ell_t\Phi_{t+1}
\ge2\kappa M^{-\alpha}\sum_{t\ge1}1=+\infty,
\]
contrary to Part 2.

Note that no step would have $\ell_t=f_t(P_t)$: such a step changes no frontier but has a positive contribution and can be deleted with a reduced total contribution, contradicting the definition of $V_0$.

By the definition of $V_0$ and the contribution bound from Part 2,
\[
V_0\le E_{F_0}(X)\le V_0.
\]
Thus $E_{F_0}(X)=V_0$, as required.
\end{proof}

\section{Terminology and notation} \label{section-notation}

\begin{tabularx}{\linewidth}{@{}L{1.46in}Y@{}}
\toprule
\textbf{Notation/term} & \textbf{Meaning}\\
\midrule
$w$ & Number of paths.\\
$h$ & Distance of the goal from the origin.\\
$\epsilon>0$ & Parameter of the input distribution.\\
$f_{\epsilon,w}$ & Input distribution: conditional distance density $\epsilon x^{-(1+\epsilon)}$, for $x\ge1$, and uniform choice of path. \\
$H$ & Per-path measure, with $dH(h)=\frac{\epsilon}{w}h^{-(1+\epsilon)}\,dh$ for $h\ge1$.\\
$\alpha=1+\epsilon$ & Exponent in the remaining-weight function.\\
$\kappa=\frac{\epsilon}{w(1+\epsilon)}$ & Scaling constant in $\phi(x)=\kappa x^{-\alpha}$.\\
$\beta=(1+\alpha)^{1/\alpha}$ & Advance-factor threshold in Proposition 2; $\beta^\alpha=1+\alpha$.\\
Stationary & Algorithm is stationary if $\Phi_{t+1}=-\phi'(\ell_t)W_t \mbox{ } \forall t$. \\

$t^+$ & Next step after $t$ that sweeps the same path $P_t$.\\
$P_t$ & Path swept at step $t$.\\
$\ell_t$ & Sweep length at step $t$, measured from the origin.\\
$f_t(P)$ & Frontier of path $P$ before step $t$.\\
$F_t$ & State before step $t$: the collection of all path frontiers.\\
$F_0$ & Initial state, with every frontier set to $1$.\\
\rowcolor{change}
$\phi(x)$ & Remaining weight on a path with frontier $x$.\\
\rowcolor{change}
$\Phi(F)$ & Total remaining weight in state $F$.\\
\rowcolor{change}
$\Phi_t=\Phi(F_t)$ & Total remaining weight right before step $t$.\\
\rowcolor{change}
$\widehat{\Phi}_{t+1}$ & Total remaining weight associated with all paths other than $P_t$ after step $t$.\\
\rowcolor{clarify}
$2\ell_t\Phi_{t+1}$ & Contribution of step $t$ to the expected competitive ratio.\\
\rowcolor{clarify}
$E_F(X)$ & Contribution of algorithm $X$, started from state $F$:
$E_F(X)=2\sum_t\ell_t\Phi_{t+1}$.\\
$W_t$ & Total (one-way) sweep length from step $t$ up to, but excluding, step $t^+$:
$W_t=\sum_{t\le j<t^+}\ell_j$.\\
$\widehat W_t$ & Total (one-way) sweep length strictly between steps $t$ and $t^+$:
$\widehat W_t=\sum_{t<j<t^+}\ell_j$.\\
\bottomrule
\end{tabularx}


\begin{thebibliography}{8}
\bibitem{BCR} R.\ A.\ Baeza-Yates, J.\ C.\ Culberson and G.\ J.\ E.\ Rawlins, Searching in the plane,
\emph{Information and Computation} \textbf{106} (1993), 234--252.
\bibitem{KRT} M.-Y.\ Kao, J.\ H.\ Reif and S.\ R.\ Tate, Searching in an unknown environment: an optimal
randomized algorithm for the cow-path problem, \emph{Information and Computation} \textbf{131}
(1996), no.\ 1, 63--79; preliminary version in \emph{Proc.\ 4th ACM--SIAM Symposium on Discrete
Algorithms} (1993), 441--447.
\bibitem{KMSY} M.-Y.\ Kao, Y.\ Ma, M.\ Sipser and Y.\ Yin, Optimal constructions of hybrid algorithms,
\emph{Journal of Algorithms} \textbf{29} (1998), 142--164; preliminary version in \emph{Proc.\ 5th ACM--SIAM
Symposium on Discrete Algorithms} (1994), 372--381.
\bibitem{MY} Y.\ Ma and Y.\ Yin, {On the cyclic assumption of an on-line search algorithm}, \emph{manuscript},
1994; earlier version of this note.
\bibitem{Schuierer} S.\ Schuierer, A lower bound for randomized searching on $m$ rays, in: \emph{Computer Science
in Perspective} (Ottmann Festschrift), Lecture Notes in Computer Science 2598, Springer,
2003, 264--277.
\bibitem{Yao} A.\ C.-C.\ Yao, Probabilistic computations: toward a unified measure of complexity, \emph{Proc.\
18th IEEE Symposium on Foundations of Computer Science} (1977), 222--227.
\end{thebibliography}
\end{document}